\documentclass[12pt]{article}%
\usepackage{eurosym}
\usepackage{bbm}
\usepackage{amsmath}
\usepackage{amsfonts}
\usepackage{amssymb}
\usepackage{graphicx}
\usepackage{pstricks}
\providecommand{\U}[1]{\protect\rule{.1in}{.1in}}
\providecommand{\U}[1]{\protect\rule{.1in}{.1in}}

\newcommand{\C}{{\mathbb C}}

\newcommand{\R}{{\mathbb R}}
\newcommand{\Z}{{\mathbb Z}}

\newcommand{\Pc}{{\mathcal P}}
\newcommand{\Hc}{{\mathcal H}}

\newcommand{\Lc}{{\mathcal L}}

\def\rpp0{\rho_{\Pc_0}}
\def\rpp1{\rho_{\Pc_1}}

\newcommand{\ba}{\begin{eqnarray}}
\newcommand{\ea}{\end{eqnarray}}
\newcommand{\bas}{\begin{eqnarray*}}
\newcommand{\eas}{\end{eqnarray*}}
\newcommand{\be}{\begin{equation}}
\newcommand{\ee}{\end{equation}}
\newcommand{\bi}{\begin{itemize}}
\newcommand{\ei}{\end{itemize}}

\newcommand{\wh}{\widehat}

\newtheorem{theorem}{Theorem}
\newtheorem{proposition}[theorem]{Proposition}

\newtheorem{lemma}[theorem]{Lemma}

\newenvironment{proof}[1][Proof]{\noindent\textbf{#1.} }{\ \rule{0.5em}{0.5em}}
\newtheorem{preremark}[theorem]{Remark}
\newenvironment{remark}{\begin{preremark}\rm}{\hfill$\Diamond$\end{preremark}}
\newtheorem{prenotation}[theorem]{Notation}

\numberwithin{equation}{section}
\numberwithin{theorem}{section}
\begin{document}

\title{{Segal-Bargmann transforms from hyperbolic Hamiltonians}}
\author{William D. Kirwin\thanks{Center for Mathematical Analysis, Geometry and Dynamical Systems, Instituto Superior T\'ecnico, University of Lisbon, will.kirwin@gmail.com}, Jos\'e  Mour\~ao, Jo\~ao P. Nunes\thanks{Department of Mathematics and Center for Mathematical Analysis, Geometry and Dynamical Systems, Instituto Superior T\'ecnico, University of Lisbon, jmourao \& jpnunes@math.tecnico.ulisboa.pt} \, and Thomas Thiemann\thanks{Lehrstuhl f\"ur Theoretische Physik III, FAU Erlangen-N\"urnberg, thomas.thiemann@gravity.fau.de }}
\maketitle

\date

\begin{abstract}
We consider the imaginary time flow of a quadratic hyperbolic Hamiltonian on the symplectic plane, apply it
to the Schr\"odinger polarization and study the corresponding evolution of polarized sections.
The flow is periodic in imaginary time and the evolution of polarized sections has 
interesting features. On the
time intervals for which the polarization is real or K\"ahler, the half--form 
corrected time evolution of polarized sections is given by unitary operators which turn out to be  equivalent to
the classical Segal-Bargmann transforms (which are usually associated to the quadratic elliptic 
Hamiltonian $H=\frac12 p^2$ and to the heat operator).  At the right endpoint of these intervals, the evolution of polarized sections is given by 
the Fourier transform from the Schr\"odinger to the momentum representation.
In the complementary intervals of imaginary time, the polarizations are anti--K\"ahler and the Hilbert space of polarized sections collapses to
 $\Hc = \{0\}$. 

Hyperbolic quadratic Hamiltonians thus give rise to a new factorization of the Segal-Bargmann transform, which is 
very different from the usual one, where one first applies a 
bounded contraction operator (the heat kernel operator), mapping $L^2$--states to real analytic functions with 
unique analytic continuation, and then one applies analytic continuation. 
In the factorization induced by an hyperbolic complexifier,
both factors are unbounded operators but their composition is, 
in the K\"ahler or real sectors, unitary. 

In another paper \cite{KMNT}, we explore the application of the above family of unitary transforms to  the definition of new holomorphic fractional Fourier transforms.
\end{abstract}

\tableofcontents

\section{Introduction}
\label{s1}

The study of equivalence of the quantizations of a symplectic manifold for different choices of polarization
is very interesting both from the physical and the mathematical points of view.
For example, the unitarizability of the Knizhnik-Zamolodchikov-Hitchin connection on bundles of conformal 
blocks continues to be a mathematical challenge (for  
recent progress, see \cite{AnGa}). These studies have close 
links with the ``complexification'' of the group of symplectomorphisms and the corresponding
geodesics on the space of K\"ahler metrics, see for example \cite{Do99a, Do99b, BLU, MN15}, play an important role in 
recent developments. 

We will be interested in the study of the unitarity of maps
associated with the imaginary time Hamiltonian flow of a quadratic hyperbolic Hamiltonian function 
(the ``complexifier'') on the 
symplectic plane,  applied to the Schr\"odinger quantization. The idea of complexifier 
Hamiltonian function was proposed in \cite{Th96} (see also \cite{Th07, GT07a, GT07b}). In geometric 
quantization, in several families of examples, 
the analytic continuation of Hamiltonian flows from real to imaginary time allows one to connect real and 
K\"ahler (or mixed) polarizations. The corresponding time evolution of polarized sections can be described by 
a map, called Kostant-Souriau-Heisenberg (KSH) map in \cite{KMN14}. To define it,  one
chooses a triple, $(H, \Pc_0, \wh {H})$, where $H$ is the complexifier Hamiltonian function, 
$\Pc_0$ is the starting polarization (the Schr\"odinger, or vertical, polarization
in the present case) and $\widehat H$ is a representation of $H$ on the initial Hilbert space 
of polarized sections $\Hc_{\Pc_0}$. 

Let $X_H$ denote the Hamiltonian vector field of $H$ and, for $\tau = it, t\in \R$, 
denote by $\Pc_{\tau}$ the imaginary time evolution of the 
initial polarization, 
$$\Pc_{\tau} = \exp(\tau \Lc_{X_H}) \, \Pc_0.$$ 
(We will show in Section \ref{s3} that for the cases examined in this paper $\Pc_{\tau}$ is well-defined. For more general cases, see \cite{MN15}.)
Given the above triple $(H, \Pc_0, \wh {H})$, 
$\exp(-t \wh H)$ is the quantization of the complexified symplectomorphism $\exp (-\tau X_H)$ on $\Hc_{\Pc_0}$ (we are chosing $\hbar=1$).
Let $\rho (H)$ be the Kostant-Souriau prequantum operator associated to $H$ (see Section \ref{s2}.) Then, 
the KSH transform at time $-it$  is given by
\be
\label{1.ksh}
U_{\tau} = \exp(+t \rho (H)) \circ \exp(- t\wh H) \, : \,    \Hc_{\Pc_0} \longrightarrow \Hc_{\Pc_{\tau}} ,
\ee
where the action of $\exp(-t \wh H)$ on  $\Hc_{\Pc_0}$ is followed by the prequantization, $\exp(+t \rho (H))$, 
of  the complex symplectomorphism, $\exp (it X_H)$, mapping $\Hc_{\Pc_0}$ to $\Hc_{\Pc_{\tau}}$.
Note that, as recalled below, the classical Segal--Bargaman transform can be described exactly in this way.  

The study of adapted complex structures on cotangent bundles of Riemannian manifolds (\cite{GS91, LS})
as the result of imaginary time evolution under the geodesic flow was initiated
in \cite{HK11} and extended to other situations in \cite{HK15, KMN13, KMN16}.
The interpretation of the factorization of the Hall coherent state transforms 
\cite{Ha94, Ha02, Dr}, which generalize the Segal--Bargman transform to the case of more general Lie groups, as the quantization of a imaginary time symplectomorphism 
followed by the prequantization of its inverse  was proposed in \cite{KMN13, KMN14}.

In the present paper, we study the imaginary time flow of an hyperbolic Hamiltonian applied
to the Schr\"odinger polarization on $\R^2$ and the corresponding KSH
transforms. The analytic continuation to imaginary time of the Hamiltonian flow is periodic. 
Therefore, the KSH map is also periodic in imaginary time. On the
time intervals for which the polarization is real or K\"ahler, 
the KSH maps will be unitary and will actually turn out to be Segal-Bargmann transforms, disguised by a new presentation as the composition of two (unbounded) factors. 
Just before the time at which the polarization becomes anti--K\"ahler, the KSH map 
is given by the Fourier transform from the Schr\"odinger to the momentum representation. In the anti-K\"ahler region,
the Hilbert space trivializes to $\Hc = \{0\}$. 

As we will see, the new factorization
of the unitary Segal-Bargmann transforms is very different
from the one achieved by the ``traditional" coherent state transforms, 
in which one first applies a bounded contraction operator, namely the heat kernel operator,
mapping $L^2$--states to real analytic functions with 
unique analytic continuation, and then one applies analytic continuation \cite{Ha06}. 
In the factorization induced by an hyperbolic complexifier
both factors are unbounded operators but their composition is unitary in the K\"ahler or real sectors.

\section{Geometric quantization and coherent state transforms}
\label{s2}

In the present section, we will briefly review relevant aspects of geometric quantization of
$(\R^2, \omega= dx \wedge dp)$ with trivial
prequantum line bundle $ L = T^* \R \times \C$, so that the non-zero constant function is a
global
trivializing section. (For a comprehensive reference on geometric quantization see, for example, \cite{Wo}.) Let the connection on $L$ be defined by $\nabla 1 = i p dx $
and the compatible Hermitian structure by $h(1, 1)=1$.

The prequantization of the observable $f \in C^\infty(\R^2)$ is given by
\ba
\label{2.pqo} \nonumber
\rho(f) &=&  \left(i \nabla_{X_f} + f\right) \otimes 1  + i \otimes \Lc_{X_f}  = i X_f  \otimes 1 - 
L_f \otimes 1 + i \otimes \Lc_{X_f} =  \\
&=&  i X_f \otimes 1- \left(p \frac{\partial f}{\partial p} - f\right) \otimes 1   + i \otimes \Lc_{X_f}, 
\ea
where $\Lc_{X_f}$ denotes the Lie derivative acting on the half-form correction.

The Schr\"odinger representation corresponds to choosing the vertical polarization
$$ 
\Pc_{Sch}=  \langle X_x  \rangle_\C =  \langle\frac{\partial}{\partial p}\rangle_\C  \, ,
$$
for which the quantum Hilbert space is
\be
\label{ee-srp}
\Hc_{\Pc_{Sch}} = L^2(\R, d  x) \otimes \sqrt{d x } ,
\ee
with inner product normalized as follows,\footnote{The factor of $\sqrt{\pi}$ is necessary to ensure 
a unitary correspondence between the Schr\"odinger and the K\"ahler polarizations to be studied below. 
This factor also appears in \cite{FMMN}.}
\be
\label{1.ip}
\langle\psi_1,  \psi_2\rangle = \sqrt{\pi} \, \int_\R \, \overline {\psi_1(x)} \, \psi_2(x) \, dx .
\ee
Let $H$ denote a real-analytic Hamiltonian function on $\R^2$. (In this paper, we will take $H$ to be an homogeneous polynomial of degree two.) 
$H$ generates complexified canonical transformations and a family of quantizations of $(\R^2,\omega)$, as follows.
Using the push-forward by the corresponding Hamiltonian flow analytically continued to complex time $\tau$, define
\be
\label{1.polit}
\Pc_{\tau} = \exp(\tau \Lc_{X_H}) \,  \Pc_{Sch}. 
\ee

Let
$$
\mathcal T=\{\tau\in \C: \Pc_{\tau} \, \mbox{is a K\"ahler polarization}\}.
$$
That is, for $\tau\in\mathcal T$, there exists a complex structure $J_\tau$ such that 
$(\R^2,\omega,J_\tau)$ is K\"ahler, with
$$
\Pc_{\tau} = T^{(0,1)}(\R^2,J_\tau).
$$
(Note that for general $H$, $\mathcal T$ could be empty. However, as we will show, for the quadratic Hamiltonians we will consider $\mathcal T$ is open and has $0$ in its boundary.)
The corresponding Hilbert space of polarized quantum states 
$\Hc_{\Pc_{\tau}}$ is
\be
\label{1.hsit}
\Hc_{\Pc_{\tau}} = \overline{\left\{\psi \in C^\infty(\R^2) \otimes \sqrt{dw} \, : \, 
\nabla_{X_w} \otimes 1  \, \left(\psi \otimes \sqrt{dw }\right)  = 0,   \,   || \psi || < \infty \right\}},
\ee
where $w=\exp(\tau X_H) \, x$ is a $J_{\tau}$--holomorphic coordinate, 
the bar denotes norm completion 
and the inner product reads
\be
\label{1.ip2}
\langle\psi_1, \psi_2\rangle = \int_{\R^2} \overline{\psi_1 (x, p)}  \, \psi_2 (x, p) \, \sqrt{\frac{dw \wedge \overline{dw}}{(-2i) dx \wedge dp}}\, dx dp   \, .
\ee
One then has a ``quantum" bundle of Hilbert spaces $\mathcal H \to \mathcal T$, with fiber $\Hc_{\Pc_{\tau}}$ over $\tau \in \mathcal T$.
In the following, we will study the unitary equivalence of the elements of this family of quantizations by ``lifting" the action of the complexified Hamiltonian flow of $H$ on the space of polarizations, given in (\ref{1.polit}), to $\mathcal H$. This will achieved by giving KSH maps, as described in the Introduction, 
$$
U_\tau: \Hc_{\Pc_{0}}=\mathcal H_{{\mathcal P}_{Sch}}\to \Hc_{\Pc_{\tau}},
$$ 
which will be shown to be unitary isomorphisms for $\tau \in \mathcal T$.\footnote{Actually, in what follows we will, for simplicity, consider only the case $\tau =it, t\in \R.$ This does not bring any essential loss of generality since for real time the KSH map is explicitly unitary.} We can think of $\Hc_{\Pc_{0}}=\mathcal H_{{\mathcal P}_{Sch}}$ as being the fiber over an extension of $\mathcal H$ to $\mathcal T \cup\{0\}$. 

The study of K\"ahler polarizations of cotangent bundles $T^*M$ through complexified Hamiltonian flows has, by now, been explored in a number of works \cite{HK11,KMN13,KMN14,MN15,EMN}. In the case of the cotangent bundle of a Lie group of compact type, such as $\R^2\cong T^*\R$ in the present paper, it is known that quantizations in K\"ahler polarizations generated by the complexified Hamiltonian flow of the energy functional 
$$
H_E(x,Y)=\frac12 \vert\vert Y\vert\vert^2, \, x\in M, Y\in T^*_xM,
$$ 
are related by coherent state transforms of Segal-Bargmann-Hall \cite{Ha94,Dr,Ha02,FMMN,KW,KMN13,KMN14}. In fact, the classical Segal-Bargmann transform corresponds to the unitary isomorphism 
\begin{eqnarray}\nonumber
L^2(\R,dx) &\stackrel{\cong}{\to}& {\mathcal HL^2}(\C,d\nu_t)\\ \nonumber
f &\mapsto & {\mathcal C}\circ e^{\frac{\Delta}{2}} f,
\end{eqnarray}
where $\Delta$ is the Laplacian for the Euclidean metric on $\R^2$, ${\mathcal C}$ denotes analytic continuation and ${\mathcal HL^2}(\C,d\nu_t)$ is the Hilbert space of holomorphic functions on $\C$ which are $L^2$ with respect to the ``averaged heat kernel measure'' $d\nu_t$ \cite{Ha94,Dr}. Note that $L^2(\R,dx)\cong {\mathcal H}_{{\mathcal P}_{Sch}}$ and that ${\mathcal H}L^2(\R^2,d\nu_t)\cong {\mathcal H}_{{\mathcal P}_\tau}$ \cite{Ha94,FMMN}.
In \cite{KMN13,KMN14}, this factorization of the Segal-Bargmann transform, and of a large family of its generalizations, was shown to correspond to a KSH map, $U_\tau$, as in (\ref{1.ksh}), arising naturally in geometric quantization without the need to introduce by hand the heat kernel measure,
$$
SB_t = U_\tau =  e^{t\rho(H_E)} \circ e^{-t\hat H_E},
$$
where $\tau = it$ and $\hat H_E= \frac{\Delta}{2}$ is the quantization of $H_E$ in the Schr\"odinger polarization. 

In this case, the positive-definite quadratic Hamiltonian $H_E$ is associated by Schr\"odinger quantization to the elliptic operator $\hat H_E$. The norm contraction given by the heat kernel operator is exaclty compensated by the norm expansion associated to the analytic continuation and $U_\tau$ is unitary. For more general choices of $H$ this will hold at most asymptotically in the limit $t\to +\infty$ \cite{KMN13,KMN14}.

In the present paper, we will examine the case when $H_E$ is replaced by an hyperbolic (indefinite) quadratic Hamiltonian function. As we will show, for some ranges of complex time $\tau$, the resulting KSH map will be unitary even though it is given by the composition of two unbounded operators. This, in particular, shows that the condition of $\exp (-t\hat H_E)$ being a contraction semigroup is not a necessary condition for the unitarity of the KSH map.

\section{Quadratic complexifiers and the Schr\"odinger polarization}\label{s3}
Consider the quadratic hyperbolic Hamiltonian 
\be
\label{3.ham}
H = \frac 12 \left(H_{11} p^2 + 2 H_{12} px + H_{22}x^2\right),
\ee
where $\det({\rm Hess}(H)) =  H_{11}H_{22} - H^2_{12}<0$. Let $\det({\rm Hess}(H))=-\alpha^2$, with $\alpha >0.$
The corresponding Hamiltonian vector field is
\be
\label{3.hvf}
X_H = H_{11} \, p \partial / \partial x -   H_{22} \,  x \partial / \partial p +   H_{12}  \, x \partial / \partial x
-  H_{12}  \, p \partial / \partial p . 
\ee
%
%The linear flow on $\R^2$ generated
%by $X_H$ is then given by the system of first order ODEs
%$$
%\left[
%\begin{array}{c}
%\dot x \\ \dot p
%\end{array}\right] =
%\left[
%\begin{array}{cc}
%H_{12} & H_{11}\\
%-H_{22} & -H_{12} 
%\end{array}\right]
%\left[
%\begin{array}{c}
%x \\ p
%\end{array}\right]:= A \left[
%\begin{array}{c}
%x \\ p
%\end{array}\right], 
%$$
%where $A\in sl(2, \R)$, Lie algebra of $SL(2, \R)$, ie the space of $2  \times 2$ traceless  real matrices. 
%Let $I$ denote the unit matrix.
%
%\begin{proposition}  [see \cite{He}, p. 149]
%\label{3.psl2} Let $A\in sl(2, \R)$ with $\det A<0$. Then
%\be
%\begin{array}{lcllc}
%\label{3.sl2}
%% e^X &=& I + X  \,  \quad  &  {\rm if}  \,  \det X &=& 0 \\
%% e^X &=& \cos(\det (X))^{1/2} \, I + \frac{ \sin(\det (X))^{1/2}}{\det (X))^{1/2}} \, X  \quad  &  {\rm if}  \,  \det X &>& 0  \\
%e^A &=& \cosh((-\det (A))^{1/2}) \, I + \frac{ \sinh((-\det (A))^{1/2})}{(-\det (A))^{1/2}} \, A
%\end{array}.
%\ee
%\end{proposition}
%

We obtain the following result for the imaginary time flow of the Schr\"odinger polarization.
\begin{proposition}
\label{3.ppol}  Let $\tau = it, t\in \R$ and
 $H$ be given by (\ref{3.ham}). The evolution of the Schr\"odinger polarization in imaginary time $\tau$ gives the 
polarization $$\Pc_{\tau} = \exp(\tau \Lc_{X_H}) \, \Pc_{Sch}$$
which is
%\bi
% \item[a)] Case with $\alpha = \det (Hess (H)) =  0$:
% \be
% \begin{array}{lclclc}
% \label{3.pol2a}
% \hbox{{\rm anti--K\"ahler}}  &{\rm  if} &  t H_{11} &<&0 \\
% \hbox{{\rm Schr\"odinger}}  &{\rm  if}  &  t H_{11} &=&0 \\
% \hbox{{\rm K\"ahler}}  &{\rm  if}  &  t H_{11} &>&0 
% \end{array}
% \ee

% \item[b)] Case with $\alpha  > 0$ (elliptic):
% \be
% \begin{array}{lclclc}
% \label{3.pol2b}
% \hbox{{\rm anti--K\"ahler}}  &{\rm  if} &  t H_{11} &<&0 \\
% \hbox{{\rm Schr\"odinger}}  &{\rm  if}  &  t H_{11} &=&0 \\
% \hbox{{\rm K\"ahler}}  &{\rm  if}  &  t H_{11} &>&0 
% \end{array}
% \ee

%\item
\be
\begin{array}{lclclc}
\label{3.pol2c}
\hbox{{\rm a)\, anti--K\"ahler}}  &{\rm  if} &  H_{11} \sin(2\alpha t) &<&0, \\
\hbox{{\rm b)\, Schr\"odinger}}  &{\rm  if}  &  H_{11} \sin( \alpha t) &=&0, \\
\hbox{{\rm c)\, K\"ahler}}  &{\rm  if} &  H_{11} \sin(2 \alpha t) &>&0, \\
\hbox{{\rm d)\, Real polarization}} \, \langle X_{ H_{12} x + H_{11} p}  \rangle &{\rm  if}  &  H_{11} \cos( \alpha t) &=&0 \, .
\end{array}
\ee
%\ei
\end{proposition}

\begin{proof}
We have $\Pc_{Sch} = \langle X_x\rangle_\C = \langle \partial / \partial p \rangle_\C$ and 
$$
\Pc_{\tau} = \exp(\tau \Lc_{X_H}) \Pc_{Sch} =  \exp(\tau \Lc_{X_H}) \langle X_x\rangle_\C  = 
\langle X_{\exp(\tau {X_H} ) (x)} \rangle_\C   .
$$
The result then follows directly from the analytic continuation of  the real time flow to imaginary time.
In real time $t\in \R$, we obtain
$$
x_t = \exp({t X_H}) (x) =
\left( \cosh(\alpha t) + H_{12} \frac{\sinh(\alpha t)}{\alpha} \right) \, x  +
 H_{11} \frac{\sinh(\alpha t)}{\alpha} \, p    \, ,
$$
and therefore in imaginary time, $\tau = it$, the function whose Hamiltonian vector field
generates $\Pc_{\tau}$ is
\be
\label{3.xit}
w = w_\tau  := x_{\tau} = \exp({\tau X_H}) (x) =
\left( \cos(\alpha t)  + i H_{12} \frac{\sin(\alpha t)}{\alpha} \right) \, x  + i
 H_{11} \frac{\sin(\alpha t)}{\alpha} \, p    \,  . 
\ee
We see that $w$ is a periodic function of $t$, with period $2 \pi / \alpha$, and  we have
$$
\frac {dw \wedge \overline{dw}}{2i}  = -\frac{H_{11}}{2\alpha} \, \sin(2 \alpha t) \, dx \wedge dp , 
$$
so that (\ref{3.pol2c}) follows. 
\end{proof}

We see that for $k\in \Z$ and $t\in (\frac{k\pi}{\alpha},\frac{k\pi}{\alpha}+\frac{\pi}{2\alpha})$, $\Pc_{\tau}$ is K\"ahler. $(\Pc_{\tau}$ is real iff $t=\frac{k\pi}{2\alpha}$.) The Proposition shows that, for appropriate intervals of $t$, namely when $\Pc_{\tau}$ is not real, $w_\tau$ defines a new holomorphic coordinate on $\R^2$ defining a new complex structure on it.

Since, from Proposition \ref{3.ppol}, the family of polarizations $\{{\mathcal P}_\tau\}, \tau =it, t\in \R$, is periodic in $t$ with period $\frac{\pi}{\alpha}$ we will consider the interval $t\in [0,\frac{\pi}{\alpha}]$ in the remainder of the paper.

\section{Unitarity of the KSH map}

\subsection{Hyperbolic case}\label{shc}

Here,  we will study  the unitarity of quantizations obtained by the action of the complexifier given by the hyperbolic Hamiltonian 
\be
\label{42.ham}
H = \frac 12 (p^2 -\alpha^2 x^2),    
\ee
where $\alpha>0.$ We are therefore setting $H_{11}=1, H_{12}=0$ and $H_{22}=-\alpha^2$. 
We will see below in Propositions \ref{4.prop} and \ref{4.prop2} that
the KSH map for
a more general hyperbolic Hamiltonian can be reduced to this one 
by a linear canonical transformation.

 The prequantization of (\ref{42.ham}) and its quantization 
in the Schr\"odinger representation are, respectively, 
\ba
\label{42.qhams} \nonumber 
\rho(H) &=& i \left(p  \frac{\partial}{\partial x} \ + \alpha^2 x  \frac{\partial}{\partial p}\right)  \otimes 1 \ - L_H   \otimes 1 + i \otimes 
\Lc_{X_H}\,  \\
\wh H &=& \frac 12 \left(- \frac{\partial^2}{\partial x^2} -\alpha^2 x^2\right) \otimes 1  
\ea
where
\be
L_H(p,x) =  \frac 12 \left(p^2  +\alpha^2 x^2 \right) .
\ee
We have from the proof of Proposition \ref{3.ppol}, with $\tau = it, t\in \R$, 
\be
\label{4.itf}
\left( 
\begin{array}{c}
p_{\tau}\\
x_{\tau}
\end{array}
\right) = 
\exp(\tau X_H) \,
 \left( 
\begin{array}{c}
p\\
x
\end{array}
\right) = 
S(\tau)
\left( 
\begin{array}{c}
p\\
x
\end{array}
\right) 
= 
\left( 
\begin{array}{rr}
\cos(\alpha t)& i\alpha \sin(\alpha t)\\
\frac{i}{\alpha} \sin(\alpha t) & \cos(\alpha t)
\end{array}
\right) \,
\left( 
\begin{array}{c}
p\\
x
\end{array}
\right) \, .
\ee

We see that, with the prequantum and quantum Hamiltonians
(\ref{42.qhams}),  and unlike the usual coherent state like transforms,  both factors in the KSH map (\ref{1.ksh})
are unbounded operators for any value of $t$. 
However, as we will see below, for ``half of the values'' in periodic imaginary time 
the ``directions of unboundedeness'' of the factors are exactly cancelled.

Recall from Proposition \ref{3.ppol} that $\Pc_{\tau}, \tau = it, t\in \R$ is periodic in $t$ with period $\pi/\alpha$, such that $\Pc_{0}$ is the Schr\"odinger polarization. For the complexifier (\ref{42.ham}) we obtain that $\Pc_{i\frac{\pi}{2\alpha}}$ is the momentum polarization. We will focus on the interval
$t\in (0,\frac{\pi}{2\alpha})$ for which ${\mathcal P}_\tau$ is a K\"ahler polarization and on the values $t=0,\frac{\pi}{2\alpha}$ for the Schr\"odinger and momentum real polarizations, respectively. For $t\in (\frac{\pi}{2\alpha},\frac{\pi}{\alpha})$, ${\mathcal P}_\tau$ is anti-K\"ahler and, as explained below, the corresponding Hilbert spaces of normalizable polarized states are trivial, ${\mathcal H}_{{\mathcal P}_\tau}=\{0\}.$ 

As in \cite{GS12, GKRS} it will be convenient to consider an overcomplete 
system of normalized (with respect to (\ref{1.ip})) Gaussian coherent states
\be
\label{42.gau}
\psi_Y(x) = {\pi}^{-\frac12} \, e^{-i P(x-Q) - \frac 12 (x-Q)^2} \otimes \sqrt{dx}   \,   ,  \qquad  Y=(P, Q) \in \R^2 .
\ee
%In fact it would be sufficient to consider real centers $Y = (P, Q) \in \R^2$ as
%a Gaussian coherent state with complex center is equal to a state
%with a real center obtained by an appropriate state--dependent projection from
%$\C^2$ to $\R^2$ (see \cite{GS12}). We will return to this
%issue in Section \ref{s5}. Let us formulate our main result.

Our main result in this Section is

\begin{theorem}
\label{42.tth}
For $\tau = it, t\in \R,$  with $t\in (0,\frac{\pi}{2\alpha})$, define the (half-form corrected) KSH transform 
\ba
\label{42.th1}
 U_{\tau} & : &  \, L^2(\R, dx) \otimes \sqrt{dx} \, \longrightarrow \,  \Hc_{\Pc_{\tau}}\\ \label{factori}
U_{\tau} &=& 
 \exp(+t \rho (H)) \circ \exp(- t\wh H)   = \,   \\   \nonumber 
 &=&  \exp \left(it \left(p  \frac{\partial}{\partial x} +\alpha^2 x  \frac{\partial}{\partial p}\right)  - t L_H(p, x)\right) 
\otimes \exp(it \Lc_{X_H})
\circ
\exp \left( \frac t2 \left( \frac{\partial^2}{\partial x^2} +\alpha^2 x^2\right) \right) \otimes 1.
\ea
The map $U_\tau$ has the following important properties:
\bi
\item[a)]  $U_\tau$ intertwines the unitary irreducible representations of the Heisenberg group
on $\Hc_{\Pc_{Sch}}$ and on $\Hc_{\Pc_{\tau}}$;
\item[b)] $U_\tau$ is a unitary isomorphism;
\item[c)] Let $Y = (P, Q) \in \R^2$. Then the $U_\tau$ transform of Gaussian coherent states reads
\ba
\nonumber {\scriptstyle \left(U_{\tau} \, \psi_Y \right) (p, x) = \pi^{-\frac12} \left| 
\frac{\sin (\theta) }{\sin(\theta+ \alpha t)}\right|^\frac12 \, e^{-\int_0^t \, \left(L_H(p_{is}, x_{is})  - 
L_H(P_{is}, Q_{is})\right) \, ds} \cdot} \\ \nonumber {\scriptstyle \cdot  
e^{-i P_{\tau}(w_\tau-Q_{\tau}) - \frac{b_\tau}{2} (w_\tau-Q_{\tau})^2}  \otimes \sqrt{dw_\tau}} = 
\ea
 \ba
\label{42.uit} 
&& {\scriptstyle =   
\pi^{-\frac12} \left| 
\frac{\sin (\theta) }{\sin(\theta+\alpha t)}\right|^\frac12 \,
e^{-i P(x-Q)}\cdot} \\ \nonumber
&& {\scriptstyle \cdot \exp\left(\frac{-\cos(\theta) \sin({\alpha t})(p-P)^2 -\sin(\theta) \cos({\alpha t})(x-Q)^2 
-2i\cos(\theta)\sin({\alpha t})(x-Q)(p-P)}{2\sin({\alpha t+\theta})}\right)}\\ \nonumber
&&{\scriptstyle \otimes \sqrt{dw_t},} \ea
where $p_{\tau}, x_{\tau}$ are given in (\ref{4.itf}), $w_\tau=x_{\tau}$ is
the $\Pc_{\tau}$--polarized coordinate (holomorphic for $t \neq \frac {k\pi}{2\alpha}, k \in \Z$), 
$P_{\tau}, Q_{\tau}$ are given by  (\ref{4.itf}) with $(p, x)$ replaced
by $(P, Q)$ and $b_\tau = \alpha\cot(\theta + \alpha t)$ with $\tan \theta = \alpha$.
\ei
\end{theorem}

\begin{remark}For  the  interval $t\in (\frac{\pi}{2\alpha},\frac{\pi}{\alpha})$, formula (\ref{42.uit}) still holds but the states on the right 
hand side  are non-normalizable.  In fact,
$\Hc_{\Pc_{\tau}} = \{0\}$  for all values $t \, :  (2k+1)\pi < 2\alpha t <(2k+2)\pi, \, k\in \Z$.
\end{remark}

\begin{remark}
Note that for the allowed interval for $t$ we always have $\theta\in (0,\frac{\pi}{2})$ and $\sin(\theta+ \alpha t)\neq 0.$
\end{remark}

For $t=\frac{\pi}{2\alpha}$ we obtain the usual relation between the quantization in the Schr\"odinger and momentum polarizations given by the Fourier transform.

\begin{theorem}
\label{42.tthreal}
For $\tau = i\frac{\pi}{2\alpha},$  the (half-form corrected) KSH transform 
\ba
\label{42.th1real}
\nonumber
 U_{i\frac{\pi}{2\alpha}} & : &  \, L^2(\R, dx) \otimes \sqrt{dx} \, \longrightarrow \,  \Hc_{{\mathcal P}_{i\frac{\pi}{2\alpha}}}\\
U_{i\frac{\pi}{2\alpha}} &=& 
 \exp(+\frac{\pi}{2\alpha} \rho (H)) \circ \exp(- \frac{\pi}{2\alpha}\wh H)   = \,   \\   \nonumber 
 &=&  \exp \left(i\frac{\pi}{2\alpha} \left(p  \frac{\partial}{\partial x} +\alpha^2 x  \frac{\partial}{\partial p}\right)  - \frac{\pi}{2\alpha} L_H(p, x)\right) 
\otimes \exp(i\frac{\pi}{2\alpha} \Lc_{X_H})
\circ
\exp \left( \frac{\pi}{4\alpha} \left( \frac{\partial^2}{\partial x^2} +\alpha^2 x^2\right) \right) \otimes 1.
\ea
coincides, up to a phase, with the Fourier transform. Namely, let $Y = (P, Q) \in \R^2$. Then the $U_{i\frac{\pi}{2\alpha}}$ transform of Gaussian coherent 
states reads
\be\label{42.uitreal} 
\left(U_{i\frac{\pi}{2\alpha}} \, \psi_Y \right) (p, x) =  
\sqrt{\frac{i}{\pi}}
e^{-i P(x-Q)}\cdot e^{-i (x-Q)(p-P)}\cdot e^{-\frac{(p-P)^2}{2}}\otimes \sqrt{dp}  = \sqrt{i}e^{-i px}{\mathcal F}(\psi_Y),
\ee
where 
\be\label{ft}
{\mathcal F}(\psi_Y) (p) = \frac{1}{\sqrt{2\pi}} \int_\R e^{ipx} {\pi}^{-\frac12}  e^{-i P(x-Q) - \frac 12 (x-Q)^2} dx \otimes \sqrt{dp}
\ee
is the Fourier transform. In particular, $U_{i\frac{\pi}{2\alpha}}: L^2(\R, dx) \otimes \sqrt{dx} \, \longrightarrow \,  \Hc_{{\mathcal P}_{i\frac{\pi}{2\alpha}}}$ is a unitary isomorphism that intertwines the usual position and momentum representations of the Heisenberg group.
\end{theorem}

\begin{remark}The phase $e^{-ipx}$ in the last term of (\ref{42.uitreal}) is expected and corresponds just to the usual change of gauge from the prequantum connection adapted to the 
Schr\"odinger polarization, which we are using, to a prequantum connection adapted to the momentum polarization.  
\end{remark}

\begin{remark}
The Fourier transform is well-known to correspond to the BKS pairing (which in this case coincides with the KSH map) between quantizations in 
the Schr\"odinger and momentum polarizations. Moreover, as shown in \cite{KW}, it arises from the usual familiy of Segal-Bargmann transforms, generated by the complexifier given by the Hamiltonian function $h(x,p)=\frac12 p^2$, in the limit $t\to +\infty$, that is when imaginary time $\tau \to +i\infty.$ It is remarkable, however, that in the present case of an hyperbolic complexifier Hamiltonian function, one obtains the Fourier transform from the KSH map {\it at finite imaginary time}.  This will be substantially clarified in 
Section \ref{2geod}.
\end{remark}

\begin{remark}\label{tttt}
Note that $U_{i\frac{\pi}{2\alpha}}$ in Theorem \ref{42.tthreal} is just the $t\to \frac{\pi}{2\alpha}$ limit of $U_{it}$ in Theorem \ref{42.tth}. Namely,
(\ref{ft}) coincides with (\ref{42.uit}) evaluated at $t=\frac{\pi}{2\alpha}$.
\end{remark}

We will prove these theorems with the help of two lemmas.
The operator $\exp(-t \wh H)$ corresponds to the Schr\"odinger evolution  
with respect to the non--Hermitian Hamiltonian
$-i \wh H$ in time $t$. The evolution of Gaussian coherent states 
 with respect to non-Hermitian quadratic Hamiltonians
was studied in \cite{GS12, GKRS}.

Let $V\subset C_\C^\infty(\R)\otimes \sqrt{dx}$ be the vector space of finite linear combinations of Gaussian coherent states of the type
$\psi_Y$ in (\ref{42.gau}).

\begin{lemma}
\label{42.lemma1}
Let $t\in \R.$
The operator $\exp(-t \wh H):V\to C_\C^\infty(\R^2)\otimes\sqrt{dx}$ is well defined and the image of the Gaussian coherent states (\ref{42.gau}) is given by
\ba
\label{42.whh} 
\left(\exp(-t \wh H) \, \psi_Y \right) (x) &=& \pi^{-\frac12} \left| 
\frac{\sin \theta }{\sin(\theta+ \alpha t)}\right|^\frac12 \,
\cdot \\ \nonumber && \cdot
e^{\int_0^t \,  L_H(P_{is}, Q_{is}) \, ds} \, 
e^{-i P_{\tau}(x-Q_{\tau})} \,  \, e^{-\frac{b_\tau}{2} (x-Q_{\tau})^2}    \otimes \sqrt{dx}, 
%\\   \nonumber
% &=& \pi^{-\frac14} \left| 
% \frac{\sin \theta }{\sin(\theta+ \alpha t)}\right|^\frac12 \,   e^{\frac{1}{4\alpha} \sin(2\alpha t) 
% \left(P^2-\alpha Q^2\right) + i \sin^2(\alpha t) PQ}\cdot  \\ \nonumber
% && \cdot\,
% e^{-i P_{\tau}(x-Q_{\tau})} \, e^{-\frac{b_\tau}{2} (x-Q_{\tau})^2}   \otimes \sqrt{dx},
\ea
where, for $\tau = it, t\in \R$,
$$
\tan \theta = \alpha,\,\, b_\tau = \alpha \cot(\theta+\alpha t),
$$
and 
$$
\int_0^t \,  L_H(P_{is}, Q_{is}) \, ds =  \frac{1}{4\alpha} \sin(2\alpha t) 
\left(P^2+\alpha^2 Q^2\right) + i \sin^2(\alpha t) PQ.
$$
\end{lemma}
\begin{proof}
By writing the initial value problem for
$$
\frac{d}{dt} \left(e^{-t\hat H}\psi_Y \right) = -\hat H \left(e^{-t\hat H}\psi_Y \right),
$$
and using as ansatz an expression as in equation (37) of \cite{GS12}, 
one obtains the expected evolution of the complex centers in complex time $\tau = it, t\in \R$,
$$
\dot P_\tau = i\alpha^2 Q_\tau,\,\,\, \dot Q_\tau = iP_\tau,
$$
and also 
$$
\dot b_\tau = \alpha^2 + b_\tau^2.
$$
The result then follows from (\ref{4.itf}) above. 
Equivalently, one can use equations (41) to (43) in \cite{GS12}. It is easy to verify that the 
(appropriately redifined) prefactor satisfies the equation (43) of 
\cite{GS12}  with the correct initial conditions and for complex time $\tau = it$. 
\end{proof}

Composing with the evolution via the prequantum operator $\rho(H)$ one obtains the following

\begin{lemma}\label{prequantumWo}
Let $\tau = it, t\in \R$, such that ${\mathcal P}_\tau$ is a K\"ahler or real polarization, that is 
$$
2\pi k \leq 2\alpha t \leq (2k+1)\pi, k\in \Z.
$$ 
Then, the KSH operator
$$
U_\tau:V\to {\mathcal H}_{{\mathcal P}_\tau}
$$
is well defined. Moreover, 
\begin{itemize}
\item[a)] The KSH--image of the Gaussian coherent states (\ref{42.gau}) is given by
\ba
\label{42.whh2} 
\left(U_{\tau}  \, \psi_Y \right) (x,p) &=&\\ \nonumber  &=& {\scriptstyle \pi^{-\frac12} \left| 
\frac{\sin (\theta) }{\sin(\theta+ \alpha t)}\right|^\frac12 \, 
e^{-\int_0^t \, \left(L_H(p_{is}, x_{is})  - L_H(P_{is}, Q_{is})\right) \, ds} \cdot} {\scriptstyle 
e^{-i P_{\tau}(w_\tau-Q_{\tau}) - \frac{b_\tau}{2} (w_\tau-Q_{\tau})^2}  \otimes \sqrt{dw_\tau}}  =\\ \nonumber
&=& {\scriptstyle \pi^{-\frac12} \left| 
\frac{\sin (\theta) }{\sin(\theta+ \alpha t)}\right|^\frac12 \,
e^{-i P(x-Q)}} \cdot 
 {\scriptstyle \exp\left({\frac{-\cos(\theta) \sin({\alpha t})(p-P)^2 -\sin(\theta) \cos({\alpha t})(x-Q)^2 
-2i\cos(\theta)\sin({\alpha t})(x-Q)(p-P)}{2\sin({\alpha t+\theta})}}\right)} \\ \nonumber
&&{\scriptstyle \otimes \sqrt{dw_t}}.
\ea

\item[b)] 
We have 
\be
\label{42.isom}
|| U_{\tau}  (\psi_Y) || = ||\psi_Y|| = 1, \qquad \forall Y =(P, Q) \in \R^2.
\ee
\end{itemize}
\end{lemma}
\begin{proof}
We have,
$$
e^{t\rho(H)}\psi = e^{-\int_0^t L_H(p(is),x(is))ds}\cdot e^{itX_H}\psi.
$$
This follows equation 8.4.2 in \cite{Wo} where, for $t\in \R$,
$$
e^{-it\rho(H)}\psi = e^{i\int_0^t L_H(p(s),x(s))ds}\cdot e^{tX_H}\psi,
$$
which can also be proved straightforwardly. Analytic continuation to $\tau = it$ gives the result.
We have to apply $\exp(t \rho(H))$ to the rhs of  (\ref{42.whh}). From (\ref{4.itf}) we see that
$$
\exp(t \rho(H)) = \left(\exp\left(- \int_0^t \, L_H(p_{is}, x_{is}) \, ds\right) \circ  
\exp\left(\tau X_H \right)\right)
\otimes \exp \left(\tau \Lc_{X_H}\right) \, .
$$
Applying this operator to  (\ref{42.whh}) we obtain (\ref{42.whh2}). It is straightforward to verify that $U_\tau(\psi_Y)$ is covariantly constant along 
${\mathcal P}_\tau$.

Let us now prove (\ref{42.isom}). Let first $t$ correspond to a K\"ahler polarization,
ie $t \, : \, 2\pi k < 2\alpha t < (2k+1)\pi, k \in \Z$.
We have, from the proof of Proposition \ref{3.ppol},
$$
- \frac i2 dw \wedge \overline{dw} = \frac{-1}{2\alpha} \sin(2\alpha t) \, dx \wedge dp.
$$
The real part of the matrix in the quadratic form in the exponent in (\ref{42.whh2}) reads,
$$
\frac 1{{2} \sin(\alpha t+\theta)} \left( \begin{array}{ll}
\cos(\theta)\sin(\alpha t) & 0 \\
0 & \sin(\theta) \cos(\alpha t)
\end{array} \right)
$$
and we obtain, using (\ref{1.ip2}), 
$$
||U_{\tau}(\psi_{Y})||^2 = \frac{\pi|\sin(\alpha t+\theta)|}{(\sin(\theta)\cos(\theta)\sin(2\alpha t))^{\frac12}} 2^{\frac12} 
\frac{\sin(2\alpha t)^{\frac12}}{{2}^{\frac12}\sqrt{\alpha}} \pi^{-1}
\frac{|\sin(\theta)|}{|\sin(\alpha t+\theta)|} =1.
$$
For the real polarization corresponding to $\cos(\alpha t) = 0$, since $H_{12}=0$, we have from 
Proposition \ref{3.ppol} that the polarization is given by 
$\langle \frac{\partial}{\partial x}\rangle_\C$. 
It is immediate to check that in this case $\vert\vert U_\tau \psi_Y \vert\vert$ is given by 
a convergent Gaussian integral along $dp$ and the above calculation also gives 
that $U_\tau$ is preserves the norms of the Gaussian states $\psi_Y, Y\in \R^2$. The same is clearly true for the Schr\"odinger polarization where 
$\sin(\alpha t)=0$ and $U_\tau$ is the identity. 
\end{proof}

We now prove Theorem \ref{42.tth}.

\begin{proof}{(of Theorem \ref{42.tth})} Part $c)$ follows directly from Lemma \ref{prequantumWo}. 
To prove $a)$ and $b)$, let $t\in (0,\frac{\pi}{2\alpha})$. 
Since the polarizations $\Pc_{\tau}$ are translation invariant
for every $t$, so that they are preserved by the Hamiltonian vector fields $X_x = -\frac{\partial}{\partial p}$ and $X_p =\frac{\partial}{\partial x}$, the corresponding Hilbert spaces of polarized 
wave functions, ${\mathcal H}_{{\mathcal P}_\tau}$, are preserved by the prequantum operators 
$\rho(x), \rho(p)$. These satisfy the Heisenberg commutation relations and one obtains an irreducible
representation   of the Heisenberg group on $\Hc_{\Pc_{\tau}}$, for every 
such value of $t$. This representation is then given by
\ba
\label{42.heis} \nonumber
V_{P_0}^{(\tau)} &=& e^{-i P_0 \rho(x)} \otimes 1 |_{\Hc_{\Pc_{\tau}}} = 
e^{-P_0 \frac{\partial}{\partial p}-iP_0x}  \otimes 1 |_{\Hc_{\Pc_{\tau}}}   \\
W_{Q_0}^{(\tau)} &=& e^{i Q_0 \rho(p)} \otimes 1 |_{\Hc_{\Pc_{\tau}}} = 
e^{-Q_0 \frac{\partial}{\partial x}}  \otimes 1 |_{\Hc_{\Pc_{\tau}}} .
\ea
The intertwining property is now an easy consequence of (\ref{42.uit}) 
and (\ref{42.heis}). On the overcomplete basis of ${\mathcal H}_{{\mathcal P}_{Sch}}$ given by the 
Gaussian coherent states in (\ref{42.gau}), we have
\bas
\left(W_{Q_0}^{(\tau)} \circ  U_{\tau}\right) (\psi_{(P, Q)})  &=&  \left(e^{-Q_0 \frac{\partial}{\partial x}}  \otimes 1 \right)
 (U_{\tau} (\psi_{(P, Q)})) =  U_{\tau} (\psi_{(P, Q+Q_0)})  = \\
&=&  \left(U_{\tau} \circ W_{Q_0}^{(0)}\right)  \,   (\psi_{(P, Q)}) \\
\left(V_{P_0}^{(\tau)} \circ  U_{\tau}\right)  (\psi_{(P, Q)})  &=&  \left(e^{-P_0 \frac{\partial}{\partial p} - i P_0x}  \otimes 1 \right)
 (U_{\tau} (\psi_{(P, Q)})) =  e^{iP_0 Q} \, U_{\tau} (\psi_{(P+P_0, Q)})  = \\
&=&  \left(U_{\tau} \circ V_{P_0}^{(0)}\right)  \,   (\psi_{(P, Q)}) ,   \quad \\
&& \forall (P,Q) , (P_0,Q_0) \in \R^2 \, {\rm and} \, \,  \forall t: \, 2k\pi \leq  2\alpha t 
\leq (2k+1) \pi, k \in \Z  .
\eas
Note, in particular, that the dense subspace $V\subset L^2(\R,dx)\otimes \sqrt{dx}$, where the 
KSH map $U_\tau$ has already been defined, is preserved by the Heisenberg group.
From the theorem of Stone-Von Neumann, this implies that the KSH map extends to a well-defined operator
$$
U_\tau:L^2(\R,dx)\otimes \sqrt{dx}\to {\mathcal H}_{{\mathcal P}_\tau},
$$ 
which moreover, must be unitary up to a multiplicative constant.
Since, from (\ref{42.isom}),
 $U_\tau$ preserves the norms of the Gaussian states it follows that 
that $U_\tau$ is, in fact, unitary for $t\in (0,\frac{\pi}{2\alpha})$. 
\end{proof}

Finally,

\begin{proof}{(of Theorem \ref{42.tthreal})} 
Equation (\ref{42.uitreal}) follows directly from Lemma \ref{prequantumWo} (See also Remark \ref{tttt}) and direct evaluation of the Fourier transform.
The proof of the fact that we have a unitary intertwining operator for the Heisenberg group
$$
U_{i\frac{\pi}{2\alpha}}: L^2(\R,dx)\otimes \sqrt{dx} \to  {\mathcal H_{{\mathcal P}_{i\frac{\pi}{2\alpha}}}},
$$
where ${\mathcal H_{{\mathcal P}_{i\frac{\pi}{2\alpha}}}}$ is the usual quantization space for the momentum polarization, 
is identical to the one for Theorem \ref{42.tth}.  
\end{proof}

\subsection{General quadratic case}

We will now see that the cases of more general hyperbolic complexifiers are equivalent to the one of the hyperbolic Hamiltonian (\ref{42.ham}) of the previous Section.
We have seen in Proposition \ref{3.ppol} that for $H_{11}=0$ the flow 
of $X_H$ leaves $\Pc_{Sch}$ invariant which implies that the KSH map
is the identity as $\wh H = \rho(H)|_{\Hc_{Sch}}$. We will therefore  assume
from now on that $H_{11} \neq 0$. 

Quadratic Hamiltonians (\ref{3.ham}) with $H_{11} \neq 0$ can
be brought to diagonal form, by a canonical transformation generated by
\be
\label{4.fab}
f(p, x) = \beta xp + \frac {\gamma}2 x^2,
\ee
for appropriate choices of $\beta, \gamma \in \R$. (See also \cite{ArGi}.) Let 
$$
H_s = \exp (s X_{f}) H,\,\,\, s\in \R.
$$

\begin{lemma}\label{eusebio}For $s=1$ and 
$$\beta=\frac12 \log|H_{11}|,\,\, 
\gamma= \frac{H_{12}}{H_{11}} \frac{\beta}{\sinh{\beta}} e^{\beta},
$$ 
we have
$$
H_1 = H_{{s}_{|_{s=1}}}= \frac12 \frac{H_{11}}{|H_{11}|} (p^2 + \det({\rm Hess}(H))x^2).
$$
\end{lemma}

\begin{proof}
From 
$$
X_{f} = \beta x \frac{\partial}{\partial x} - (\beta p + \gamma x) \frac{\partial}{\partial p},
$$
we obtain,
$$
x_s = \exp (sX_{f})x = e^{s\beta}x,\,\, p_s=  \exp (sX_{f})p = -\frac{\gamma}{\beta}
\sinh{(s\beta)} x + e^{-s\beta} p,
$$
from which the result follows.
\end{proof}

Notice that, in general, from the fact
that 
\be
\label{4.ht}
H_{1} = \exp(X_f) H
\ee
one can not, in general, deduce unitarity of
the KSH map for  $H_{1}$ from the unitarity of the KSH map 
for $H$. 

\begin{lemma}Let  $\wh H, \wh{{H_1}}, \wh f$, respectively, denote the self-adjoint Weyl quantizations of 
$H, H_{1}, f$ in the Schr\"odinger representation. Then,
\be
\label{4.qht}
\wh{{H_1}} = \exp(i \wh f) \circ \wh H \circ \exp(-i \wh f) .
\ee
\end{lemma}
\begin{proof}
It is well known that Weyl quantization is exact for observables which are quadratic in $x$ and $p$ (see, 
for example, \cite{U} and \cite{AdPW}) that is $\hat{\{f,g\}}=i[\hat f, \hat g]$ for quadratic $f,g$. But then,
$$
\frac{d}{ds}H_s=\{f,H_s\}
$$
implies
$$
\hat H_s = \exp(i s\wh f) \circ \wh H \circ \exp(-i s\wh f) ,
$$
since this is the only solution of 
$$
\frac{d}{ds}\hat H_s = i [\hat f, \hat H_s],\,\, \hat H_0=\hat H.
$$
\end{proof}

It follows that
\bas
\exp(-s \wh{{H_1}}) =  \exp(i \wh f) \circ \exp(-s\wh H) \circ \exp(-i \wh f).
\eas

We also have 
\begin{lemma} The prequantum operators $\rho(H)$ and $\rho(H_1)$ satisfy
\bas
\exp(s \rho({{H_1}})) &=&  \exp(i \rho(f)) \circ \exp(s\rho(H)) \circ \exp(-i \rho(f)).
\eas
\end{lemma}

\begin{proof}
Consider the following prepotential for $-\omega$, 
$$
\tilde \theta = \frac12 (pdx - xdp) = \theta -\frac12 d(xp),
$$
where $\theta = pdx$.
Note that this corresponds to a different choice of trivializing section for the prequantum bundle, so that the 
expressions of wave functions in both gauges are related by
$$
\psi_{\tilde \theta} = e^{\frac{i}{2}xp}\psi_\theta.
$$
In this new frame, it is easy to verify that the prequantum operator for $f$ is
$$
\tilde \rho(f) = \left(iX_f-\tilde\theta(X_f)+f\right)\otimes 1 + i\otimes {\mathcal L}_{X_f} 
= iX_f\otimes 1 + i\otimes {\mathcal L}_{X_f}.
$$
It is then immediate to verify that
$$
\tilde \rho(H_s) = e^{s{\mathcal L}_{X_f}} \tilde\rho(H) e^{-s{\mathcal L}_{X_f}},
$$
from which the result follows by rewriting the prequantum operators in the original gauge.
\end{proof}

Note that the operators $\exp(i \wh f), \exp(i \rho(f))$ are unitary. 
Then, for the KSH maps we obtain
\ba
\label{4.scon} 
U_{is}^{{{H_1}}} &=& \exp(s \rho({{H_1}})) \circ \exp(-s \wh{{H_1}}) \\  \nonumber
&=&  \exp(i \rho(f)) \circ \exp(s\rho(H)) \circ \exp(-i \rho(f))  \circ  \exp(i \wh f) \circ \exp(-s\wh H) \circ \exp(-i \wh f) , 
\ea
and  we obtain the following sufficient condition for the unitarity of
the KSH map for $H_1$.
\begin{proposition}
\label{4.prop}
Let the conditions (\ref{4.ht}) and (\ref{4.qht}) hold. Then if 
\be
\label{4.suf}
\rho(f)|_{\Hc_{\Pc_{Sch}}} = \wh f \otimes 1 
\ee
 the unitarity of the KSH map
for $H$ is equivalent to the unitarity of the KSH map for $H_1$. 
\end{proposition}

\begin{proof}
Substituting (\ref{4.suf}) in (\ref{4.scon}) we obtain
\bas
 U_{is}^{{{H_1}}}  &=& \exp(i \rho(f)) \circ \exp(s\rho(H)) \circ \exp(-t\wh H) \circ \exp(-i \wh f)  = \\
&=&  \exp(i \rho(f)) \circ   U_{is}^{{H}} \circ \exp(-i \wh f)  
\eas
and therefore the operator $ U_{is}^{{{H_1}}}$ is unitary if and only if the operator 
$ U_{is}^{{{H}}}$ is unitary.
\end{proof}

\begin{proposition}
\label{4.prop2}
The function $f$ in (\ref{4.fab}) satisfies the condition
(\ref{4.suf}).
\end{proposition}

\begin{proof}
We have
$$
\Lc_{X_{f}} \, dx = \beta dx , 
$$
and therefore, taking into account the half-form correction,
\ba \nonumber
\rho(f)|_{\Hc_{\Pc_{Sch}}} &=& \left[ \left(iX_{f} + \frac{\gamma}2 x^2\right) \otimes 1
+i \otimes \Lc_{X_{f}}\right]|_{\Hc_{\Pc_{Sch}}} \\
&=&  \left[i \beta x \frac{\partial}{\partial x} +\frac i2 \beta + \frac {\gamma}2 x^2\right] \otimes 1 = \wh{f}.
\ea
The last equality corresponds to the (Weyl) symmetric ordering and satisfies (\ref{4.qht}) for all quadratic $H$.
\end{proof}

\begin{remark}
Notice that quadratic functions, $f= \delta p^2 + \beta px + \gamma x^2$,  satisfy (\ref{4.suf}) if and only if
$\delta =0$. This is why one can not prove unitarity of the KSH map for e.g. $p^2$ or $p^2-x^2$ from the
corresponding trivial unitarity of the KSH maps for $x^2$ and $xp$, respectively.  As we mentioned above,  in the latter cases the KSH 
maps are equal to the identity.
\end{remark}

%\subsection{$\alpha = 0$ and elliptic cases}

\begin{remark}Given the Propositions \ref{4.prop} and \ref{4.prop2} the unitarity of the KSH map 
for quadratic observables (\ref{3.ham}) with $H_{11} \neq 0$ is reduced
to the unitarity of the KSH map  for one of the Hamiltonians: 
%\bi
%\item[(i)]
%$\frac {p^2}2 (\alpha =0)$
%\item[(ii)] $\frac 12 (p^2+x^2), \frac 12(p^2-x^2)$
%\ei
\be
\label{4.canc}
\begin{array}{llll}
(i) & \frac {p^2}2   \\
(ii) & \frac {p^2+\alpha^2 x^2}2,      \\
(iii) & \frac {p^2-\alpha^2 x^2}2  , \alpha\neq 0, 
\end{array}
\ee
where, with no loss in generality, we have taken $H_{11}>0$.
The KSH for the free particle (i) in (\ref{4.canc}) coincides with 
the Segal-Bargman-Hall transform and is known to be unitary. The elliptic case, $(ii)$, is treated in
\cite{Es}.
\end{remark}

\section{Relation to the classical Segal-Bargmann transform}
\label{2geod}

Consider the Hamiltonian for a free particle on the symplectic plane
$$
H_E = \frac12 p^2,
$$
with Hamiltonian vector field 
$$
X_{H_E}= p\frac{\partial}{\partial x}.
$$

The corresponding Hamiltonian flow in imaginary time $\tilde \tau = i\tilde t, \tilde t\in \R,$ produces the complex structure $\tilde J_{\tilde \tau}$ on the plane given by the following holomorphic coordinate
$$
\tilde w_{\tilde \tau}= {\rm exp}\,(\tilde \tau X_{H_E}) (x) = x + i\tilde t p.
$$

Even though the Hamiltonian flows of $H_E$ and of the hyperbolic Hamiltonian $H=\frac12 (p^2-\alpha^2x^2)$ are quite different, by comparing with (\ref{3.xit}) we see that this holomorphic coordinate correponds to the same complex structure on the plane as defined by the Hamiltonian flow of $H$ in complex time $\tau = it, t\in \R$, such that
\be
\label{timerep}
\tan (\alpha t) = \alpha \tilde t.
\ee

In fact, if $\tau, \tilde \tau$ satisfy this relation we have 
$$
w_\tau = \cos(\alpha t) x + \frac{i}{\alpha} \sin (\alpha t) p = \cos(\alpha t) \tilde w_{\tilde \tau},
$$
so that these two holomorphic coordinates on the plane, where one is just a rescaling of the other, define the same complex structure, that is $J_\tau = \tilde J_{\tilde \tau}$.

Note that, therefore, the same K\"ahler structure (that is, up to equivalence) on the symplectic plane can be described in two different ways, as we now explain. Recall 
that the path of K\"ahler  structures obtained by 
Hamiltonian flow in imaginary time is a geodesic path with respect to the Mabuchi affine connection\footnote{In the noncompact case of the plane one does not have a Mabuchi metric but the corresponding affine connection on the space of K\"ahler potentials still makes sense.} defined on the space of K\"ahler potentials \cite{Do99a,MN15}. 
The velocity of the geodesic is related to the Hamiltonian function generating it. 

To the Hamiltonian flows of $H$ and $H_E$ in imaginary times, $\tau = it, \tilde \tau = i\tilde t, \,t,\tilde t\in \R_+$, respectively, we can associate two families of diffeomorphisms of the symplectic plane 
$$
\Phi_\tau (x,p) = (\cos(\alpha t) x, \frac{1}{\alpha} \sin(\alpha t)p),\,\,\, \tilde \Phi_{\tilde \tau} (x,p) =(x, \tilde t p).
$$
If $J_{st}$ denotes the standard complex structure on the plane, corresponding to the usual holomorphic coordinate $z=x+ip$, we then have maps
$$
(\R^2,\omega_\tau,J_{st})\stackrel{\Phi_\tau}{\rightarrow} (\R^2,\omega,J_\tau)
$$
and
$$
(\R^2,\omega_{\tilde\tau},J_{st})\stackrel{\tilde \Phi_{\tilde\tau}}{\rightarrow} (\R^2,\omega,\tilde J_{\tilde\tau}).
$$

Here, on the right-hand side, $J_\tau, \tilde J_{\tilde\tau}$ are the complex structures described above and which together with the original symplectic form $\omega$ define K\"ahler structures on the plane. This is the ``symplectic picture'' where $\omega$ is kept fixed during the evolution in imaginary time. The Moser maps $\Phi_\tau, \tilde \Phi_{\tilde\tau}$ pull-back these to equivalent K\"ahler structures in the ``complex picture'' where the complex structure $J_{st}$ is kept fixed \cite{Do99a,MN15}.

As we noted above, for $\tan (\alpha t) = \alpha \tilde t$ we have $J_\tau = \tilde J_{\tilde\tau}$. In fact, for $\tau,\tilde\tau$ satisfying this relation it is easy to check that 
$\Phi_\tau^{-1}\circ \tilde \Phi_{\tilde\tau}$ is just a rescaling of the $(x,p)$ coordinates which is an automorphism of $J_{st}$. 

Remarkably, we therefore have the same Mabuchi geodesic described by the flow of $H_E$ in imaginary time $\tilde\tau$ for $\tilde t\in (0,+\infty)$, or by the flow of $H$  in imaginary time $\tau$ for $t\in (0,\frac{\pi}{2\alpha})$. This is possible due to the fact that the complex structures these flows generate have continuous groups of automorphisms as we described above.

The lift of the imaginary time Hamiltonian flow of $H$  to the bundle of Hilbert spaces of quantum states is given by the KSH maps $U_\tau$
 that were studied in Section \ref{shc}. In the case of the free particle Hamiltonian $H_E$, this lift is well-known to correspond to a family of  Segal-Bargmann transforms
 \cite{Ha02, FMMN, KW}.  
The fact that these two Hamiltonians generate the same Mabuchi geodesic, together with the intertwining properties of the corresponding KSH maps relative to the Heisenberg group, implies that the lifts of their imaginary time Hamiltonian flows to the bundle of quantum states are the same. Therefore, the KSH maps $U_\tau$ are really Segal-Bargmann transforms disguised by the realization of the Mabuchi geodesic in which it is generated by the hyperbolic Hamiltonian. Let $SB_{\tilde\tau}, \tilde\tau = i\tilde t, \tilde t>0$ denote the family of Segal-Bargmann transforms of \cite{Ha02, FMMN, KW} which can be written as \cite{KMN13, KMN14}
 $$
 SB_{\tilde \tau} = {\rm exp}\,(+\tilde t\rho(H_E)) \circ {\rm exp}\,(-\tilde t{\widehat H_E}),
 $$
 where $\tilde \tau = i\tilde t, \tilde t\geq 0$, the prequantum operator associated to $H_E$ is 
 $$\rho(H_E)=\left( (iX_{H_E}-H_E)\otimes 1 + i\otimes {\mathcal L}_{X_{H_E}}\right)$$  and ${\widehat H_E} = -\frac12 \frac{\partial^2}{\partial x^2}.$
 Notice that $SB_{\tilde \tau}$ can be obtained from Theorem \ref{42.tth} by taking $\alpha\to 0$ in the expression for $U_{\tilde\tau}$. The classical Segal-Bargmann transform 
 is obtained for $\tilde t=1.$
 
 We then have the remarkable operator identity

 \begin{theorem}\label{kmnt} \cite{KMNT}
 For $\tilde \tau = i\tilde t, \tilde t\geq 0$, let $\tau =it, t\in [0,\frac{\pi}{2\alpha})$, be such that $\tan{(\alpha t)}=\alpha \tilde t.$ 
 Then,
 $$
 SB_{\tilde \tau}  =U_{\tau} .
 $$ 
 Therefore, the Segal-Bargmann transforms can be factorized as the composition of the unbounded time-evolution operators associated to the prequantum and Schr\"odinger quantizations of the hyperbolic Hamiltonian $H$, 
 as in (\ref{factori}).
 \end{theorem}
 
 At the endpoints of the intervals of imaginary time, $t\to\frac{\pi}{2\alpha}$ and $\tilde t \to +\infty$ respectively, the transforms $U_\tau$ and $SB_{\tilde \tau}$ become the usual Fourier transform. For the family $U_\tau$ this is Theorem \ref{42.tthreal} above and for the family $SB_{\tilde\tau}$ this is proved in \cite{KW}. The application to holomorphic fractional Fourier transforms, as well as the proof of Theorem \ref{kmnt}, will appear in \cite{KMNT}.

%\newpage

\section*{Acknowledgements}
We would like to thank T.Baier for useful discussions.
The authors JM and JPN were partially
by 
FCT/Portugal through the projects UID/MAT/\-04459/2013, PTDC/MAT-GEO/3319/2014, PTDC/MAT-OUT/28784/2017 and by the (European Cooperation in Science and 
Technology) COST Action MP1405 QSPACE. The authors also 
thank the generous support from the Emerging Field Project on Quantum Geometry from Erlangen--N\"urnberg University, 
where this project was initiated.

\providecommand{\bysame}{\leavevmode\hbox to3em{\hrulefill}\thinspace}

\end{document}